\documentclass[11pt]{article}
\usepackage[left=1in,right=1in,top=1in,bottom=1in]{geometry}
\usepackage{amsmath,amsthm,stmaryrd,amsfonts}
\usepackage{tikz}
\usetikzlibrary{shapes.geometric, arrows}
\usepackage{subfigure}
\usepackage{float}
\usepackage{hyperref}
\usepackage{authblk}
\usepackage{natbib}

\newtheorem{theorem}{Theorem}
\newtheorem{lemma}[theorem]{Lemma}
\newtheorem*{remark}{Remark}

\providecommand{\keywords}[1]
{
  \small	
  \textbf{\textit{Keywords---}} #1
}

\begin{document}
%----------------------------------------------------------------------------------------
%	TITLE SECTION
%----------------------------------------------------------------------------------------
 % Article title
\title{Optimal liquidation with temporary and permanent price impact, an application to cryptocurrencies}
%%\tnotetext[t1]{This work was supported by the Scientific Colombia Program – EFI Alliance: programs and policies for the promotion of a formal economy, code 60185, which are part of the EFI Alliance – Formal and Inclusive Economy, under the Contingent Recovery Agreement No. FP44842-220-2018.}
%Authors
\author[1]{Hugo E. Ramirez \thanks{\href{mailto:hugoedu.ramirez@urosario.edu.co}{hugoedu.ramirez@urosario.edu.co}}}

\author[1]{Juli\'an Fernando S\'anchez L\'opez \thanks{\href{mailto:julianf.sanchez@urosario.edu.co}{julianf.sanchez@urosario.edu.co}}}
%\equalcont{These authors contributed equally to this work.}

\affil[1]{Universidad del Rosario, Calle 12C No. 4-69, Bogot\'a, Colombia}

%\date{\today} 

%----------------------------------------------------------------------------------------

\maketitle

\begin{abstract}
This paper studies the optimal liquidation of stocks in the presence of temporary and permanent price impacts, and we focus in the case of cryptocurrencies. We start by presenting analytical solutions to the problem with linear temporary impact, and linear and quadratic permanent impact. Then, using data from the order book of the BNB cryptocurrency, we estimate the functional form of the temporary and permanent price impact in three different scenarios: underestimation, overestimation and average estimation, finding different functional forms for each scenario. Using finite differences and optimal policy iteration, we solve the problem numerically and observe interesting changes in the optimal liquidation policy when applying calibrated linear and power forms for the temporary and permanent price impacts. Then, with these optimal policies, we identify optimal liquidation trajectories and simulate the liquidation of initial inventories to compare the performance among the optimal strategies under different parametrizations and against a naive strategy. Finally, we characterize the optimal policies based on the functional form of the inventory and find that policies generating the highest revenue are those starting with a low trading rate and increasing it as time passes.
\end{abstract}

\keywords{Optimal liquidation, price impact, finite differences, Stochastic Control, cryptocurrencies}

\section{Introduction}
The market of cryptocurrencies is increasing in importance since the net value of the market has grown from $\$3.8 $ billion $USD$ in $2015$ to $\$943.41$ billion $USD$ in the second quarter of $2022$, having a maximum of over $ \$3000$ Billion $USD$\footnote{Data from https://www.statista.com/statistics/730876/cryptocurrency-maket-value/ and  https://coinmarketcap.com/}. Likewise, the use of cryptocurrencies has advantages such as its articulation with Decentralized Finance (DeFi)\footnote{In accordance with \citet{pi22} DeFi can be understood as the process that makes use of the blockchain for the development and implementation of novel financial products and services.}, security and transparency, among others, which, according to \citet{blan21}, make them very convenient investment alternatives. However, this type of investment is very volatile and many investors may have or want to acquire or liquidate their positions optimally on a given time period. In this paper, we focus on optimal strategies for liquidation of assets considering price impacts, specifically different forms of temporary and permanent price impact calibrated from cryptocurrencies. 

Optimal liquidation consists in finding the optimal strategy an agent must follow when selling large amounts of shares over a finite period of time and minimizing adverse effects, consequences of her own actions. The time dependence is of great importance because, if the agent executes in a short time, by selling fast, the rate of execution increases and so does the impact on the price of shares. But if she sells at a slow pace over a long period of time, holdings are exposed to a greater uncertainty due to the volatility, which is crucial in the case of cryptocurrencies. Thus, in our setting, the agent must sell a number $\Upsilon$ of shares in a given time frame $[0,T)$, where she must find the optimal selling rate $ \nu_t $ such that the revenue from liquidation is maximized. To this end, we adopt an optimization problem based on a model that includes linear and non-linear impact functions, and is calibrated with data from the Limit Order Book (LOB) of the BNB cryptocurrency. To tackle the problem, we use stochastic control techniques and present two original analytical solutions in specific scenarios, but in more general settings and when analytical solutions were not at hand, we use numerical methods, that is finite differences over a non-linear Partial Differential Equation (PDE). 

The well-known model for optimal execution of transactions in \citet{ac00} establishes two types of impact for the asset price: temporary and permanent. In their work, temporary impact refers to ``...temporary imbalances in supply in demand caused by our trading leading to temporary price movements away from equilibrium'' and permanent impact ``...means changes in the equilibrium price due to our trading which remain at least for the life of our liquidation''. Theoretical approaches agree that the permanent price impact must be linear, \citet{g10} shows this condition is necessary to avoid dynamic arbitrage. He presents \textit{the principle of no-dynamic-arbitrage} which states that the trading cost is non-negative for any round-trip trade strategy, i.e. $\displaystyle \int_{0}^{T}\nu_t dt=0$, and shows that non-linear permanent market impact is inconsistent with the principle of no-dynamic-arbitrage. Conversely, our study focuses in only liquidating shares, and then the sum of this sequence of trades cannot be zero, therefore the principle of no-dynamic arbitrage is not applicable. Consequently our model considers non-linear permanent price impacts, which are estimated from data, and studies the effect of such impact functions on the optimized trading rate.

Additionally, empirical evidence shows that market impact is not linear, but similar to a square root. That is, it is proportional to the square root of the volume of shares executed, as seen in \citet{teb16} and \citet{Almgren05}, or to the square root of the trade duration, as in \citet{br13}. Moreover, different approaches using nonlinear temporary or permanent price impact have been studied in the literature, for example \citet{g14} relates to the result in \citet{g10} and extends some theoretical results to the nonlinear setting, \citet{bl18} offers a theoretical approach, considering stochastic temporary and permanent price impacts. \citet{a} proposes a model where the impact depends on the theoretical shape of LOB given by a density (exponential), whereas our model does not rely on any shape and is thus capable of capture the actual book in simulations. More recently \citet{BRUNOVSKY2018} solves for optimal liquidation with linear price impact, although importantly shows some theoretical advances by adding the possibility of acquiring shares as well as selling.

Although related literature shows different perspectives to the solution of the optimal liquidation problem for large amounts of shares, in this paper we study the problem based on the setting proposed in \citet{cjp15}, but using more general functional forms for temporary and permanent price impact (TPI and PPI respectively). At first, we extend the theoretical model in \citet{cjp15} by considering linear TPI, and linear and quadratic PPI, and show close form solutions for these models. Since analytical solutions are scarce, for more general models we use numerical methods, specifically finite differences and optimal policy iteration, and solve the problem for different forms of the TPI and PPI functions, i.e. functions result of calibration. This type of problems, i.e. a PDE with two dimensions in space and only one diffusion, are frequently numerically solved using the ADI scheme, but have some stability issues, as shown in the classical work of \citet{Peaceman55}, henceforth also, for example, in \citet{Thomas98},\citet{Tavella00} and more recently by \citet{duffy04}. So in order to improve convergence in the multi-dimensional finite differences we propose the use of an alternative scheme, namely implicit directional scheme. 

After that, using the technique in \citet{cj16} we calibrate, in a cryptocurrency market, the functional forms of the TPI and PPI functions. A first aspect to highlight is that calibrated parameters of the functional forms for the TPI and PPI change according to the depth in ticks and therefore the volumes of the LOB \footnote{We refer to the volumes of the LOB as the number of shares posted at each tick level}. That is, given a fixed LOB setting (depth in ticks and volumes), if the inventory to liquidate is sufficiently small, the liquidation can be done barely walking the book, so that the impact would be small and increasing as the liquidation target is reached, this configures a power form $x^a$, where $a>1$. The opposite case occurs when the inventory to be liquidated is large for the LOB setting, since at first the impact would be large, but decreasing as the liquidation continues, this configures a power form $x^a$ but having $0<a <1$. As a third scenario, we find a in-between point, where the TPI and PPI approaches a linear form, i.e. $x^a$ with  $a=1$. Using this three configurations for calibration gives different functions for TPI and PPI, we are able to identify the different optimal trading policies for each of these scenarios. Thus, in accordance with \citet{Lillo03} and \citet{Almgren05}, modeling the problem must include power forms for TPI and PPI, but for these more complex and general scenarios, the use numerical techniques to find solutions is compulsory.

We find that optimal liquidation strategies change for different TPI and PPI functions, and that these outperform a naive strategy of liquidating all the inventory at a single point in time, which we confirm with simulations. We observe that, we may describe these strategies in a common functional form of the optimal inventory, as a non-increasing monotonic function $\displaystyle q(t) = \frac{-\Upsilon}{T^{d_2}}t^{d_2} +\Upsilon$ subject to $q(0) = \Upsilon$ and $q(T) = 0$. Under certain conditions and using real life data simulations, similar to a paper trading, the optimal policies starting with a slow rate and increasing it as time passes have the highest cumulative revenue. Thus, our results in some sense confirm the conclusions about concavity of market impact in \citet{Curato2015}, stating that concave functions generate more profit. This suggests that beyond the results in calibration, the performance of optimal policies does not depend strongly on the functional form of TPI and PPI, and that such performance may be characterized by the initial amount of shares to be liquidated compared with the availability of the LOB. This has an important potential in practice, since eliminates the hustle of calibrating models and thus saves time by implementing directly the optimal strategy.

This paper is structured as follows: section $2$ introduces the model, based on \citet{cjp15}, and analytical solutions of the problem of liquidation from the perspective of optimal stochastic control. Section $3$ presents analytical solutions when the TPI is linear, and PPI is linear and quadratic as extensions of \citet{cjp15}, and a general solution using finite differences, for more general TPI and PPI. Section $4$ shows the estimation of a possible functional forms of TPI and PPI and to calibration of the parameters for the BNB cryptocurrency, using the LOB data. In section $5$, we find the optimal policies and analyze them by varying some of the calibrated parameters. Also, we compare the performance of numeric and naive trading strategies simulating the optimal liquidation of initial inventories on LOB data. Finally, we establish a characterization of the optimal policies from the optimal inventory available, and independent of the calibration. In section $6$ we give some final remarks.

%------------------------------------------------
\section{The model}
We work on a probability space $\left(\Omega, \mathcal{F},\mathbb{P} \right)$ and build on the model proposed by \citet{cjp15}, the agent has an amount $\Upsilon$ of shares of an asset that she wants to liquidate in the time interval $\left[0,T\right)$ and obtain the most profit. The inventory changes through time and has dynamics ${dQ}^\nu_t={- \nu}_tdt$, thus $\nu_t$ is the trade rate. The mid price of the asset is governed by the SDE 
\begin{equation}
\displaystyle {dS}^\nu_t=- g\left(\nu_t\right)dt+\sigma dW_t,
\end{equation}
where $W_t$ is a standard Brownian motion, $\sigma \in \mathbb{R}^+$ represents the volatility and $g: \mathbb{R}^+ \to \mathbb{R}^+$ is the permanent price impact, dependent of the trade rate. When an agent sells shares of the stock the actual price of the transaction is called the execution price and follows 
\begin{equation}
\displaystyle {\hat{S}}^\nu_t= \left(S^\nu_t - \frac{1}{2}\Delta \right) - f\left(\nu_t\right),
\end{equation}
where $\Delta$ is the bid-ask spread in the LOB and the function $f: \mathbb{R}^+ \to \mathbb{R}^+$ is the temporary price impact. Finally the agent wants to maximize her final expected utility wealth, earnings from transactions, that is
\begin{equation}
\displaystyle \mathbb{E}\left\{\int^T_0{{\hat{S}}^\nu_t\nu_tdt}\right\}
\end{equation}
the rate $\nu_t$ is controlled by the agent at each time point $t$ and her actions affect the asset's liquidation price, thus, the optimization problem can be represented by the agent's value function
\begin{equation}\label{eqn:AgentValueFunc}
H\left(t,S,q\right) = \sup_{\nu_t \in  \mathcal{A}(t,T)} { \mathbb{E}\left\{\int^T_t \left(\left(S^\nu_r - \frac{1}{2}\Delta \right)-f\left(\nu_r\right)\right)\nu_r dr \right\}}
\end{equation}
where $\mathcal{A}(t,T)$ represents the admissible set of non-negative bounded strategies and the equation \eqref{eqn:AgentValueFunc} satisfies the $\mathcal{H}\mathcal{J}\mathcal{B}$ partial differential equation (PDE) \footnote{For a detailed explanation, refer to chapter $6$ of \citet{cjp15}}
\begin{equation} \label{eqn:AgentHJB}
\displaystyle 0 = {\partial }_tH+\frac{1}{2}{\sigma }^2{\partial }_{ss}H +\sup_{\nu \in  \mathcal{A}}\left\{-{g(\nu)\partial }_sH-\nu{\partial }_qH+\left(S-\frac{1}{2}\Delta -f(\nu)\right)\nu\right\}
\end{equation}
subject to
\begin{eqnarray}
	\label{eqn:cond1HJB}
	\displaystyle H(T,S,q)& \rightarrow & -\infty,  \text{ when } t \rightarrow T \text{ and }q>0,\\
	\label{eqn:cond2HJB}
	\displaystyle H(t,S,0)& \rightarrow & 0, 
\end{eqnarray}
the first condition corresponds to the penalty, negative revenue, for approaching T with positive inventory, and the second is complementary and guarantees that at time $t$ the revenue is zero if the inventory is zero.

Next, we extend the model in \citet{cjp15} by changing the PPI to a linear and quadratic function of the trade rate $\nu$. Although linear PPI is treated in \citet{cjp15} and \citet{cj17}, our approach is essentially different because of the type of penalty, in the range $[0,T)$. Since our penalty is a boundary condition going to $-\infty$, we are not allowing the manager to have any inventory at maturity. On the other hand, \citet{cjp15} and \citet{cj17} have a quadratic penalty, which, although costly, allows the manager to end up with some inventory.

\section{Model solution approaches}
\subsection{Analytical approach}
It is always satisfactory to have analytical solutions, this section shows the scenarios where we were able to extend a close form formula for solutions. The following lemmas show the analytical solutions for the  $\mathcal{H}\mathcal{J}\mathcal{B}$ PDE \eqref{eqn:AgentHJB} subject to conditions \eqref{eqn:cond1HJB} and \eqref{eqn:cond2HJB}, when PPI is linear and quadratic, respectively. These lemmas may be seen as extensions of the solution given in section $6.3$ of \citet{cjp15}.
\begin{lemma}[\textbf{Linear permanent price impact}]\label{lm:linear}
	Let TPI and PPI be linear functions, that is $f\left(\nu_t\right)=a_1 \nu_t +a_2\ $, and  $g(\nu_t)=b_1\nu_t$ \footnote{For simplicity and to find an analytical solution of $\mathcal{H}\mathcal{J}\mathcal{B}$ PDS \eqref{eqn:AgentHJB} we use $g(\nu_t)=b_1\nu_t$ instead of $g(\nu_t)=b_1 \nu_t + b_2$ since the independent term $b_2$ does not allow us to find an analytical solution.} respectively, then the solution for the  $\mathcal{H}\mathcal{J}\mathcal{B}$ differential equation \eqref{eqn:AgentHJB} subject to \eqref{eqn:cond1HJB} and \eqref{eqn:cond2HJB} is
	\begin{equation}
		\label{eqn:Hlin}
		H\left(t,S,q\right)=q\left(S-\frac{1}{2}\Delta-a_2\right)-\left(\frac{b_1}{2} +\frac{a_1}{T-t} \right)q^2,
	\end{equation}
	and the optimal strategy for the liquidation of $q_t$ shares in the time interval $[t,T)$ is
	\begin{equation}
		\label{eqn:Nulin}
		\nu_t^*=\frac{q_t}{T-t}. 
	\end{equation}  
\end{lemma}
\begin{proof}
	See appendix \ref{sec:App1}.
\end{proof}
\begin{lemma}[\textbf{Quadratic permanent price impact}]\label{lm:linqu}
	Let us have a linear TPI and a quadratic PPI, that is $f\left(\nu_t\right)=a_1 \nu_t +a_2$, and  $g(\nu_t)=c_1\nu_t^2+c_2\nu_t\ $, respectively, then the analytical solution for the  $\mathcal{H}\mathcal{J}\mathcal{B}$ PDE \eqref{eqn:AgentHJB} subject to \eqref{eqn:cond1HJB} and \eqref{eqn:cond2HJB} is
	\begin{equation}
		\label{eqn:Hquad}
		H\left(t,S,q\right)= \left\{ \begin{array}{lcc}
			q\left(S-\frac{1}{2}\Delta-a_2\right)-\frac{c_2}{2}q^2-\frac{4\left(c_1q+a_1\right)^3}{9c_1^2(T-t)} &   if  & q \neq 0 \\
			\\ 0 & if & q = 0,
		\end{array}
		\right.
	\end{equation}
	and the optimal strategy for the liquidation of $q_t$ shares in the time interval $[t,T)$ is
	\begin{equation}
		\label{eqn:Nuquad}
		\nu_t^*=\frac{2\left(c_1q_t+a_1\right)}{3c_1\left(T-t\right)}
	\end{equation}  
\end{lemma}
\begin{proof}
See appendix \ref{sec:App2}.
\end{proof}
\begin{remark}
As in lemma \ref{lm:linear} no independent term is included in the quadratic PPI function, because this allows us calculate an analytical solution of the equation \eqref{eqn:AgentHJB}.
\end{remark}
From lemmas \ref{lm:linear} and \ref{lm:linqu} we appreciate that whereas $H$ is quadratic in the $q$ variable for the former, it is cubic for the latter, although linear for $S$ in both cases. This implies that the change of the PPI from linear to quadratic affects in a similar manner the optimal revenue $H$. Although interestingly note that the optimal policy in both cases is a linear function of $\displaystyle \frac{q}{T-t}$.   

Since obtaining analytical solutions to any type of PPI function is a hard task, and in many cases even impossible, to extend our model, we use numerical solutions. Specifically, we want to test optimal strategies in a cryptocurrency market and use, besides linear, calibrated power TPI and PPI functions for which we do not have analytical solutions.

\subsection{Numerical approach}\label{sec:Numerics}
\subsubsection{\textit{Finite differences}}
We use \textit{implicit finite differences} to find the solution of the PDE \eqref{eqn:AgentHJB}. This technique approximates the partial derivatives through discrete expressions constructed from Taylor expansions and thereby giving an approximate solution of the differential equation.

We denote the discretized version of the function $H$ by $H_{i,j}^k=H\left(k\Delta t,i\Delta s,j\Delta q\right)$, where $t=k\Delta t$, $s=i\Delta s$ and $q=j\Delta q$, additionally we set $S_{max}=N_s\Delta s$, $Q_{max}=N_q \Delta q$ and $T=N_t\Delta t$. As boundary conditions are great of importance in this technique, we state these next:
\begin{itemize}
\item  When the price is $S=0$, we think of it as an absorbent state and the possible revenue for potential trades at this price is $0$, thus
\begin{equation}
\displaystyle H(t,0,q) = 0.
\end{equation}
\item If the inventory reaches the $q=0$ level, there is no more shares to sell, having
\begin{equation}
H(t,S,0)=0.
\end{equation}
\item If the price reaches a maximum (we have to fix a maximum price $S_{max}$ for the numerical approach) the revenue is the amount of shares sold multiplied by the fixed price. That is the inventory multiplied by the fixed maximum price.
\begin{equation}
\displaystyle H(t,S_{max},q)=S_{max}q.
\end{equation}
\item We avoid the case of having any inventory $q>0$ at final time $T$, thus we introduce a heavy penalty (minus infinity to be sure it never happens)
\begin{equation}
H(T,S,q)= -\infty.
\end{equation}
\end{itemize}
By using central differences on $S$, forward differences on $q$ and backward differences on $t$, we obtain the following discretized differential equation that must be optimized
\begin{equation}
\begin{split}
\frac{H_{i,j}^{k+1}-H_{i,j}^{k}}{\Delta t}+\frac{\sigma^2}{2}\left(\frac{H_{i+1,j}^{k}-2H_{i,j}^{k}+H_{i-1,j}^{k}}{(\Delta s)^2}\right) -\nu \left(\frac{H_{i,j}^{k}-H_{i,j-1}^{k}}{\Delta q}\right) \\ -g(\nu)\left(\frac{H_{i+1,j}^{k}-H_{i-1,j}^{k}}{2\Delta s}\right) +\left(i\Delta s-\frac{1}{2}\Delta -f(\nu)\right)\nu=0,
\end{split}
\end{equation}
which has approximation of orders $O(\Delta t),O(\Delta s^2), O(\Delta q)$.

By grouping similar terms and leaving the unknown terms (time $k\Delta t$ and space $\geq j\Delta q$) on the same side of the equality we obtain
\begin{equation}\label{eq:tridiagonal}
AH_{i-1,j}^{k}+BH_{i,j}^{k}+CH_{i+1,j}^{k} = Z_{i,j}^{k*} 
\end{equation}
where 
\[Z_{i,j}^{k*} =DH_{i,j}^{k+1} + E H_{i,j-1}^{k} + F_i(\nu)\]
 \[F_i(\nu)=-\left(i\Delta s-\frac{1}{2}\Delta -f(\nu)\right)\nu\]
and $A$,$B$,$C$,$D$,$E$ are the corresponding constants.

Since there is no diffusion in the $q$ variable the traditional $ADI$ scheme for $2$ dimensions could lead to oscillations and instabilities. That is why we prefer this approach, namely \textit{Directional $2$-dimensional finite differences}, where we use implicit solutions for vectors changing the variable $S$ but having the variable $q$ constant, which additionally simplifies to tridiagonal systems. More precisely, for each time $k \Delta t$ before expiry, we begin by calculating a boundary condition at $q=0$, and once solved for $q = (j-1) \Delta q$, we step to $q = j \Delta q$ and use the tridiagonal system in \eqref{eq:tridiagonal} to solve for all the values of $S$, and iterate in ascending order over $q$ up to $Q_{max}$. That is, on a uniform grid, we find the values of the function $H$ for all values of $i$ while iterating in ascending order over $j$ and then stepping backwards through time in $k$, as shown in figure (\ref{fig:uniformgrid}).
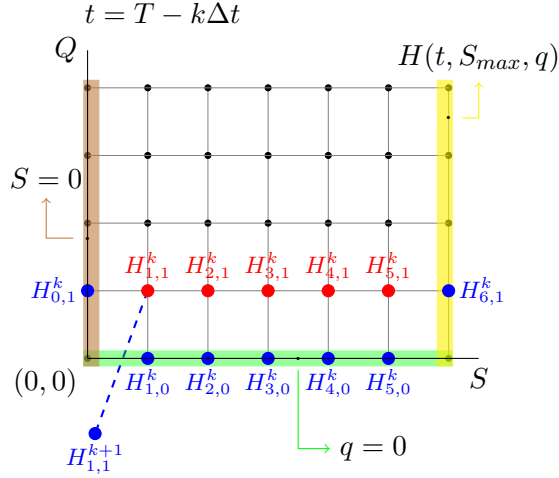
\begin{figure}
	\centering
	\begin{tikzpicture}
	\foreach \x in {-2.4,-1.6,-0.8,...,3.2}
	{
		\draw[very thin, gray] (\x,-1.8) -- (\x,1.8);
		\foreach \y in {-1.8,-0.9,...,2.7}
		{
			\filldraw (\x,\y) circle (0.04);
		}
	}
	\foreach \y in {-1.8,-0.9,...,2.7}
	{
		\draw[very thin, gray] (-2.4,\y) -- (2.4,\y);
	}
	\draw (-2.4,-1.8) node[below left] {$(0,0)$};
	\draw[blue,thick,dashed] (-1.6,-0.9) - - (-2.3,-2.8);
	\filldraw[blue](-2.3,-2.8) circle (0.08) node[below] {\footnotesize $H_{1,1}^{k+1}$};	
	\filldraw[color=green,opacity=0.4] (-2.45,-1.9) -- (2.45,-1.9) -- (2.45,-1.7) -- (-2.45,-1.7) --cycle;
	\filldraw[color=brown,opacity=0.6] (-2.45,-1.9) -- (-2.45,1.9) -- (-2.25,1.9) -- (-2.25,-1.9) --cycle;
	\filldraw[color=yellow,opacity=0.6] (2.45,-1.9) -- (2.45,1.9) -- (2.25,1.9) -- (2.25,-1.9) --cycle;	
	\filldraw[blue](-1.6,-1.8) circle (0.08) node[below] {\footnotesize $H_{1,0}^k$};
	\filldraw[blue](-0.8,-1.8) circle (0.08) node[below] {\footnotesize $H_{2,0}^k$};
	\filldraw[blue](0.0,-1.8) circle (0.08) node[below] {\footnotesize $H_{3,0}^k$};
	\filldraw[blue](0.8,-1.8) circle (0.08) node[below] {\footnotesize $H_{4,0}^k$};
	\filldraw[blue](1.6,-1.8) circle (0.08) node[below] {\footnotesize $H_{5,0}^k$};
	\filldraw[blue](-2.4,-0.9) circle (0.08) node[left] {\footnotesize $H_{0,1}^k$};
	\filldraw[red](-1.6,-0.9) circle (0.08) node[above] {\footnotesize $H_{1,1}^k$};
	\filldraw[red](-0.8,-0.9) circle (0.08) node[above] {\footnotesize $H_{2,1}^k$};
	\filldraw[red](0.0,-0.9) circle (0.08) node[above] {\footnotesize $H_{3,1}^k$};
	\filldraw[red](0.8,-0.9) circle (0.08) node[above] {\footnotesize $H_{4,1}^k$};
	\filldraw[red](1.6,-0.9) circle (0.08) node[above] {\footnotesize $H_{5,1}^k$};	
	\filldraw[blue](2.4,-0.9) circle (0.08) node[right] {\footnotesize $H_{6,1}^k$};
	%nodes
	\node (A) at (-2.4,-0.2){.};
	\node (B) at (-2.95, 0.6) {$S=0$};
	\node (C) at (0.4,-1.8) {.};
	\node (D) at (1.4, -3.0) {$q=0$};
	\node (E) at (2.4,1.4) {.};
	\node (F) at (2.8,2.2) {$H(t,S_{max},q)$};
	% arrows
	\draw[brown][->, to path={-| (\tikztotarget)}]
	(A) edge (B);
	\draw[yellow][->, to path={-| (\tikztotarget)}]
	(E) edge (F);
	\draw[green][->, to path={|- (\tikztotarget)}]
	(C) edge (D);
	
	\draw[black] (-1.4, 2.5) node[above] {$t=T-k\Delta t$};
	\draw[] (-2.4,-1.8) -- (-2.4,2.3) node[left] {$Q$};
	\draw[] (-2.4,-1.8) -- (2.8,-1.8) node[below] {$S$};
	\end{tikzpicture}
	\caption{This figure shows an example of the uniform grid at time $t=T-k\Delta t$, where the brown highlighted line corresponds to the boundary at $S=0$, the yellow to the boundary at $S=S_{max}$ and the green to the boundary at $q=0$. The blue dots correspond to known values and the red dots correspond to unknown values that will be found simultaneously at each step through $q$ using the tridiagonal system \eqref{eq:tridiagonal} } \label{fig:uniformgrid}
\end{figure}
\subsubsection{\textit{Optimization}}
For the optimization process, since $f(\nu)$ is the TPI and $g(\nu)$ is the PPI, we apply optimal policy iteration using first order condition (FOC) on the control $\nu$, that is
\begin{equation}
\label{eqn:FOC_v}
-g'(\nu)\partial_sH-\partial_qH+S-\frac{1}{2}\Delta-\left(f'(\nu)\nu+f(\nu)\right)=0
\end{equation}
Specifically:
\begin{itemize}
\item For linear TPI and quadratic PPI, that is $f\left(\nu\right)=a_1 \nu +a_2\ $, and $g(\nu)=c_1\nu^2+c_2\nu+c_3$ the FOC is
\begin{equation}
	-\left(2c_1\nu+c_2 \right)\partial_sH-\partial_qH+S-\frac{1}{2}\Delta-\left(2a_1\nu+a_2 \right)=0.
\end{equation}
\item For power TPI and power PPI, that is $f\left(\nu\right)=r_1 \nu^{r_2} +r_3\ $, and $g(\nu)=p_1\nu^{p_2}+p_3$ the FOC is
\begin{equation}\label{eqn:FOC_vp}
	-p_1p_2\partial_sH\nu^{p_2-1}-\partial_qH+S-\frac{1}{2}\Delta-\left(r_1(r_2+1)\nu^{r_2}+r_3\right)=0.
\end{equation}
\end{itemize}
Note that linear cases are included in the power form just by taking $r_2=1$ in $f(\nu)$ and/or $p_2=1$ in $g(\nu)$.

\section{Model calibration}
This section presents the implementation of a method, based on \citet{cj16} and \citet{cjp15}, to estimate functional forms of TPI and PPI. Additionally, to estimate the volatility, we use the realized volatility, i.e, $\hat{\sigma} =\sqrt{ \sum_{t=1}^{T}\left(\ln(S_t)-\ln(S_{t-1})\right)^2 }$, which is consistent with the model.

\subsection{\textit{Temporary Price Impact}}
To find the functional form of the TPI, following the ideas from \citet{cjp15}, we use LOB data with $n$ ticks of depth taken at $N$ intervals of $\tau$ seconds. In each of those $N$ intervals, we simulate the  execution of different sizes of Sell Market Orders, i.e. we liquidate the $m$ volumes (As orders of different sizes) $Q=\{Q_1,Q_2,...Q_i,...Q_m\}$ against the existing LOB, walking the book when necessary. The LOB is represented, for each available tick $j$, by a price $S_j$ and its corresponding number of Limit Orders $Q_{S_j}$. These volumes are chosen in an increasing order up to a significant maximum relative to the LOB's total capacity (all posted limit orders to buy). To find the execution price we define it as the weighted average of the execution prices for each volume, in bid side, that is
\[ \displaystyle \hat{S}^{bid}_{i,t} = \frac{\sum_{j=1}^{N_i} S_j Q_{s_j}^* }{Q_i}, \quad \text{ where } Q_{s_j}^* = \begin{cases} Q_i - \sum_{k=1}^{j-1}Q_{s_{k}} & \text{if } Q_i \leq \sum_{k=1}^{j}Q_{s_{k}} \\ Q_{s_j} & \text{else.} \end{cases},\]
and $\displaystyle N_i = \min \left\{j, \sum_{k=1}^{j} Q_{s_k} \geq Q_i  \right\}$, and $0 \leq t \leq N_t \tau$. Since we perform this procedure every $\tau$ seconds, we consider not only volumes, but also liquidation rates ($\nu$) in our simulations, i.e. volumes per $\tau$ seconds. Thus, by letting $S_t^{bid}$ be the best price on the bid side, the TPI at time $t$ due to the liquidation rate $\displaystyle \nu_{i} = \frac{Q_i}{\tau}$ is $\displaystyle  TPI^{bid}_{i,t} = \hat{S}^{bid}_{i,t}-S^{bid}_{t}$. Finally, we use the average of $TPI^{bid}_{i,t} $ over time, denoted by $TPI^{bid}_{i}$, versus $\nu_i$ to fit the models $TPI^{bid}(\nu)=f(\nu)=a_1\nu+a_2$ and $TPI^{bid}(\nu)= f(\nu)=r_1\nu^{r_2}+r_3$, via least squares.

\subsection{\textit{Permanent Price Impact}}
To estimate a functional form of the permanent price impact, we use the same idea as for the temporary impact, i.e. going through the data every $\tau$ seconds and at each of those moments executing simulated Market Sell Orders of different volume sizes and therefore different trading liquidation rates. Then we capture the impact of each trading rate, in this case, on the midprice $\displaystyle S^{mid}_{t} = \frac{S^{bid}_{t}+S^{ask}_{t}  }{2}$, that is, the difference between the mean price existing before the liquidation of $Q_i$ shares in $\tau$ seconds and the mean price remaining after the liquidation. Using least squares, we estimate the parameters of the models $PPI^{bid}(\nu)=g(\nu) =b_1 \nu+b_2$ and $PPI^{bid}(\nu)= g(\nu) = p_1\nu^{p_2}+p_3$.\\

%------------------------------------------------

\section{Results}

\subsection{Model parameters calibration}
We use LOB intraday data with $n=100$ ticks of depth and a frequency of $\tau=5$ seconds for $N_t=1440$ (two hours) taken from Binance API for BNB \footnote{BNB is the official cryptocurrency of the Binance exchange and powers the BNB Chain ecosystem} on February 6, 2022.

A first aspect to consider in the calibration is to capture, from the LOB, the total available volume at each moment that could be liquidated, i.e., for each moment add up the available volumes (number of shares per tick) of the 100 ticks. This information is very important in the calibration of the functional form of both the TPI and the PPI, because if the amount of shares to be liquidated over a time lapse $\tau$ is small compared to the available volumes in the LOB, the impact will be underestimated, since it would barely walk the book. In this case, the calibration generates a power functional form $(x^\alpha \text{ with } \alpha>1)$, we call this scenario underestimation. In the opposite case, if the amount of shares to be liquidated is higher than available volumes, the impact will be overestimated, since transactions have to walk the book. Moreover, after a certain volume the full depth of the LOB will be reached and therefore the impact will not grow more, this leads to a power functional form $(x^\alpha \text{ with } 0 <\alpha < 1)$, we call this scenario overestimation. Finally, the intermediate case is found thinking of liquidation rates, we see the effect on $\nu = Q/ \tau$ since liquidation volumes are considered in a time lapse $\tau$. Thus, for underestimation scenario, we consider $\nu_{max} = 50$, for overestimation scenario $\nu_{max} = 7000$, and as an intermediate value between the two previous cases, we use $\nu_{max} = 1200$. This intermediate scenario leads to a calibration of the functional forms of TPI and PPI very close to linear function $(x^\alpha \text{ with } \alpha=1)$, we call this scenario average-estimation.

The underestimated scenario is shown in figure \ref{fig:underimpact}, where subfigure (\ref{subfig:underTPI}) shows the plot of trading rates $\nu_i$ vs $TPI^{bid}_i$, for volumes $1 \leq i \leq m$, as well as the plots of the linear and power calibrated models $TPI^{bid}(\nu)=f(\nu) =a_1 \nu+a_2$ and $TPI^{bid}(\nu)= f(\nu) = r_1\nu^{r_2}+r_3$. Subfigure (\ref{subfig:underPPI}) shows plots of trading rates $\nu_i$ vs. $PPI^{bid}_i$ and the calibrated PPI models $PPI^{bid}(\nu)=g(\nu) =b_1 \nu+b_2$ and $PPI^{bid}(\nu)= g(\nu) = p_1\nu^{p_2}+p_3$. We can see that the shape of the average TPI and PPI function can be calibrated more closely through a power function than through a linear function. Due to the above, we consider, to find optimal policies in this scenario, power $(x^\alpha \text{ with } \alpha>0 \text{ and } \alpha \neq 1)$ and linear $(x^\alpha \text{ with } \alpha=1)$ impact forms.

\begin{figure}
	\centering
	\subfigure[TPI]
	{\label{subfig:underTPI}\includegraphics[scale=0.42]{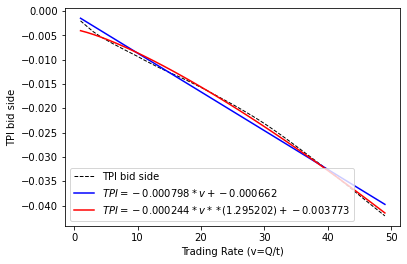}}
	\subfigure[PPI]
	{\label{subfig:underPPI}\includegraphics[scale=0.42]{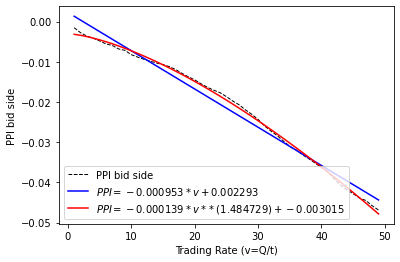}}
	\caption{Underestimation of power and linear functional forms of TPI and PPI with $\nu_{max} = 50$, $m=50$} 
	\label{fig:underimpact}
\end{figure}
In figure \ref{fig:overimpact}, we have the overestimated scenario and specifically note the changes in the concavity of the calibrations.
\begin{figure}
	\centering
	\subfigure[TPI]
	{\includegraphics[scale=0.42]{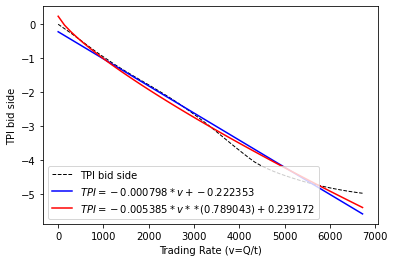}}
	\subfigure[PPI]
	{\includegraphics[scale=0.42]{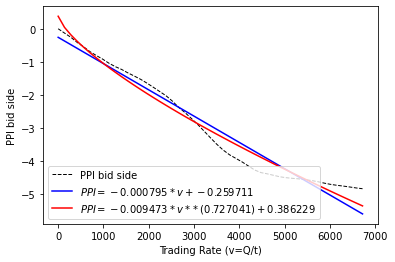}}
	\caption{Overestimation of power and linear functional forms of TPI and PPI with $\nu_{max} = 7000$,$m=50$} 
	\label{fig:overimpact}
\end{figure}
Additionally, the average-estimation scenario is in figure \ref{fig:meanimpact}.
\begin{figure}[H]
	\centering
	\subfigure[TPI]
	{\includegraphics[scale=0.42]{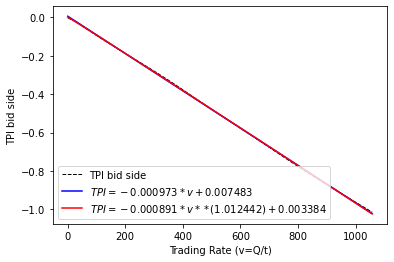}}
	\subfigure[PPI]
	{\includegraphics[scale=0.42]{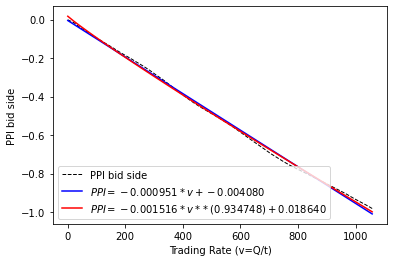}}
	\caption{Average-estimation of power and linear functional forms of TPI and PPI with $\nu_{max} = 1200$,$m=50$} 
	\label{fig:meanimpact}
\end{figure}

To measure which of the two functional forms, linear or power, fits best the average TPI or PPI on the different scenarios, we use the squared residuals, defined as $\hat{\epsilon}^2_{i,t} = \left(y_i - \hat{y}_i \right)^2_t$ and the corresponding R-squared, where $y_i$ is $TPI^{bid}_{i}$ or $PPI^{bid}_{i}$ and $\hat{y_i}$ is $TPI^{bid}(\nu_i)$ or $PPI^{bid}(\nu_i)$ (respectively). In appendix \ref{sec:App4}, we show these results and we observe that, in all cases, both the total sum of the residuals and the average of the residuals are lower for the power functional form than for the linear form. Likewise, the $R^2$ is higher in all cases for the power functional form than for the linear form. This implies a better fit in all cases for the power functional form, even in the average scenario where the differences are minimal. We also obtain an average spread of $\Delta = 0.100069$ and an estimated constant volatility of $\sigma=0.009388$.

\subsection{Optimal policies}
This section refers to the algorithm optimizing the trading rate, that is the numerics for solving the $\mathcal{H}\mathcal{J}\mathcal{B}$ PDE and using the optimal policies in \eqref{eqn:FOC_v} and \eqref{eqn:FOC_vp}.

\subsubsection{\textit{Error and Convergence analysis}}
To verify our numerical algorithm, we perform some tests, first we calculate the absolute error (when feasible) comparing the analytical and numerical value function $H$ at $t=0$, that is for linear and quadratic PPI, obtaining an accuracy of $10^{-5}$ and $10^{-4}$ respectively.

Additionally, we were able to confirm the convergence rate of the numerical method for each of the independent variables of the value function $H$, i.e. inventory $q$, price $s$ and time $t$, finding that the convergence rate is $1$ for $q$ as expected. As for the price $s$, convergence rate is very close to the theoretical value $2$ for all points of the grid, likewise, the convergence rate for time $t$ is around its theoretical value $1$.

Finally, to verify that there are no significant differences between the theoretical vs numerical optimal policies $\nu$, we applied the Chi-square goodness of fit test between numerical and analytical policies for the linear and quadratic PPI. We verify there are no significant differences between these policies, i.e. in essence the policy found through the numerical method is the same as the one found analytically. After performing these tests, we are confident that our algorithm is correct and is working properly.

\subsubsection{\textit{Analysis of optimal policies}}
Next, we show the optimal policies obtained using the numerical approach in section \ref{sec:Numerics}, and we explore the results of using the calibrated linear and power TPI and PPI for the scenarios under, over and average estimation. 

Figure (\ref{fig:3Dopt_nu_Under05pp}) shows the plot of optimal policies along with a path for the policy and inventory beginning at $0.5$ with power TPI and PPI in underestimation scenario. To construct the path following the optimal policy and the optimal inventory, starting at $q_0$ inventory, with the corresponding optimal trading rate $\nu(0,q_0)$, we find $q_1 = q_0 - \nu(0,q_0)\Delta t$. If $q_1$ lies on the grid, the procedure continues, otherwise the respective value of $\nu(1,q_1)$ is interpolated. We repeat the procedure using $q_k = q_{k-1} - \nu(k-1,q_{k-1})\Delta t$ from $t=0$ to $t=T=1$. Figure (\ref{fig:2Dopt_nu_Under05pp}) shows the corresponding 2D graph of the policy path. Note that the figure (\ref{fig:3Dopt_nu_Under05pp}) is a plot of $\nu(t,q)$, i.e. not depending on $S$, because we verify that optimal policies are independent of the price, as we can see in figure (\ref{fig:nuindep}).

\begin{figure}
	\centering
	\subfigure[3D optimal policy for fixed inventory $q=0.5$]
	{\label{subfig:tsindep}\includegraphics[scale=0.42]{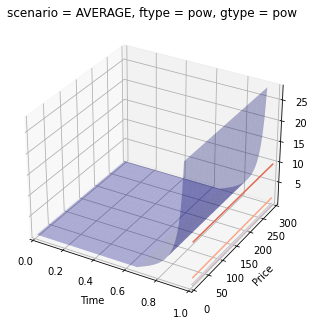}}
	\subfigure[3D optimal policy for fixed time $t=0.5$]
	{\label{subfig:qsindep}\includegraphics[scale=0.42]{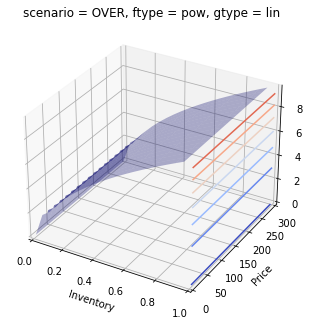}}
	\caption{Optimal policies for fixed inventory and time respectively. We observe that the policy is price-independent in all scenarios and under the different combinations of the TPI and PPI functions. The graphs show two examples of this situation.} \label{fig:nuindep}
\end{figure}

By changing the starting point of the inventory to $0.2$, figures (\ref{fig:3Dopt_nu_Under02pp}) and (\ref{fig:2Dopt_nu_Under02pp}) show 3D optimal policy, optimal policy path, optimal inventory path and corresponding 2D optimal policy path respectively. We observe some changes in the liquidation policy, since starting with inventory $0.5$ implies liquidation at rates close to $0$ after approximately time $t=0.9$, while when starting with inventory $0.2$ this low rates are reached around $t=0.6$. One possible interpretation of this situation is that, in an underestimation scenario, a small impact causes the agent to liquidate early without much impact and avoid the risk of exposure. We opted to show only these graphs because we think are interesting since power TPI and PPI have best fit. However, as the functional forms of TPI and PPI change, also does the policy.

\begin{figure}
	\centering
	\subfigure[3D optimal policy, optimal policy path and optimal inventory path]
	{\label{fig:3Dopt_nu_Under05pp}\includegraphics[scale=0.42]{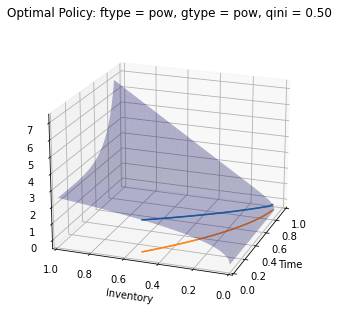}}
	\subfigure[2D optimal policy path]
	{\label{fig:2Dopt_nu_Under05pp}\includegraphics[scale=0.42]{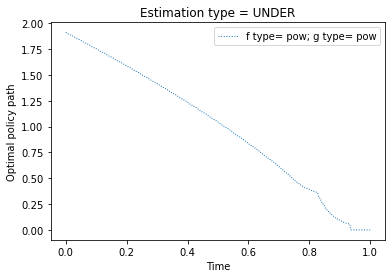}}
	\caption{Optimal policy for constant $S$ along with a path for the policy and inventory when it begin at $0.5$ with power TPI and PPI in underestimation scenario} 
\end{figure}
\begin{figure}
	\centering
	\subfigure[3D optimal policy, optimal policy path and optimal inventory path]
	{\label{fig:3Dopt_nu_Under02pp}\includegraphics[scale=0.42]{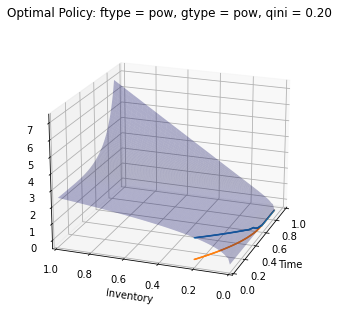}}
	\subfigure[2D optimal policy path]
	{\label{fig:2Dopt_nu_Under02pp}\includegraphics[scale=0.42]{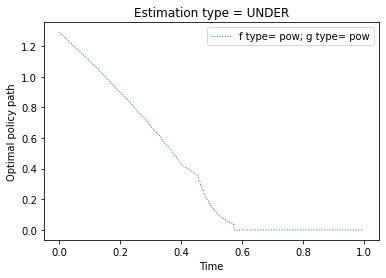}}
	\caption{Optimal policy for constant $S$ along with a path for the policy and inventory when it begin at $0.2$ with power TPI and PPI in underestimation scenario} 
\end{figure}

Using parameters $T = 1$, $S_{max} = 300$, $Q_{max} = 1$,$N_t=360$, $N_s = 10$, $N_q=100$, $\Upsilon = 0.5$, the calibrated coefficients showed in appendix \ref{sec:App3} and using the mid point price $S=150$ as a constant, we plot the optimal policy paths. In figure \ref{fig:opt_nu_comp_calUnder}, we observe that if PPI is linear, the policy path suggests to liquidate more slowly near initial time, even zero, to increase the liquidation rate to a maximum near time $T$, in which the inventory is zero, then the trading rate drop to zero. However, when PPI is power, the policy starts liquidating quickly and decreasing the trading rate over time, except towards the end of the liquidation period, because having positive inventory causes the trading rate to increase slightly until all inventory is liquidated. This may be caused by the fear of the heavy penalty imposed in the boundary to avoid positive inventory at final time $T$.

We notice a huge difference between the policies, for different TPI and PPI functions, and notably when $g$ changes from linear to power, because the policy changes from increasing to decreasing. This may be caused by the concavity induced in the underestimation scenario, since not walking the book implies a faster execution and as there are less shares to liquidate the trading rate decreases. However, this only happens when PPI is power, the opposite happens when PPI is linear, which might suggest that in this scenario linear PPI has a greater impact on the optimal policy, than power PPI.

Next we find the optimal policy paths for over and average-estimation of linear and power TPI and PPI. Figures \ref{fig:opt_nu_comp_calAver} and \ref{fig:opt_nu_comp_calOver} show a comparison of the optimal policy paths for different TPI and PPI. Interestingly, in the average-estimation scenario we observe, that the optimal policy paths are similar if PPI is power, regardless of the functional form of TPI, begins by not liquidating and then increases until a certain point at which it drops to zero.

For policies where PPI is linear, we observe that they start with a high liquidation rate, decreasing over time and quickly reach a zero trading rate. This is the opposite of the situation observed in the underestimation scenario. This suggests that in the average-estimation scenario, power PPI has a greater impact than linear PPI.

Similar situation is presented in the  overestimation scenario, in figure \ref{fig:opt_nu_comp_calOver}, only that in such scenario we observe that policy paths that liquidate early (Linear PPI) do so more quickly than those liquidating late (power PPI). Moreover, by having even a small inventory, the penalty imposed towards the end of the time period causes a very large liquidation rate, which happens when PPI is power.

According to the above, the optimal policy obtained is highly sensitive to the model calibration process and therefore the optimal policy depends on the calibration scenario, the functional form of both TPI and PPI and the LOB data, since model coefficients depend entirely on this information.
\begin{figure}
	\centering
	{\includegraphics[scale=0.5]{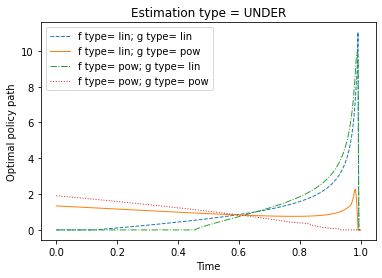}}
	\caption{Graphical comparison of optimal policies from underestimated TPI and PPI, with $\Upsilon = 0.5$} 
	\label{fig:opt_nu_comp_calUnder}
\end{figure}
\begin{figure}
	\centering
	{\includegraphics[scale=0.5]{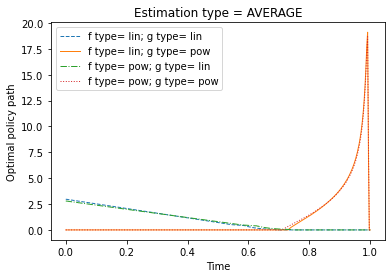}}
	\caption{Graphical comparison of optimal policies from average-estimated TPI and PPI, with $\Upsilon = 0.5$} 
	\label{fig:opt_nu_comp_calAver}
\end{figure}
\begin{figure}
	\centering
	{\includegraphics[scale=0.5]{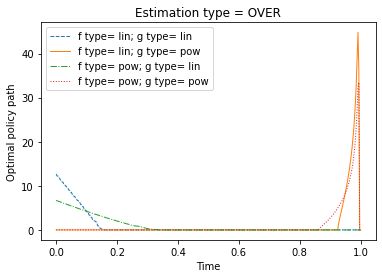}}
	\caption{Graphical comparison of optimal policies from overestimated TPI and PPI, with $\Upsilon = 0.5$} 
	\label{fig:opt_nu_comp_calOver}
\end{figure}

\subsection{Simulation}
This section presents simulations of the liquidation of $\Upsilon$ shares of BNB using the optimal policies calculated from the previous section. The simulation consists of taking an initial inventory $\Upsilon$ and liquidating that inventory over time $[0,T]$ walking the book when needed, following each of the policy paths and then compare the behavior of the numerical optimal strategies versus the naive strategy (NS), which consists of liquidating all cryptocurrencies in a single prespecified time.

To refer to the numerical optimal strategies we use the following notation: the first letter refers to the estimation scenario, i.e. under (U), over(O) and average(A), the second letter refers to TPI and therefore will always be T, the third letter refers to the functional form of TPI, i.e. linear (L) or power (P), the fourth letter will always be P in reference to PPI and the fifth and last letter refers to the functional type of PPI, i.e. L or P. 
For example, the optimal policy obtained in the underestimation scenario with linear TPI and power PPI is represented by the letters UTLPP. With these references, we use the policy paths starting with inventory $0.5$ for UTLPL, UTLPP, UTPPL, UTPPP, OTLPL, ATLPL, ATLPP, ATPPL, ATPPP, OTLPP, OTPPL and OTPPP.

In these simulations we assume, in the sense of \citet{cj17}, that the order book is resilient, i.e. when executing a MO, at the next moment in time the book has again cleared that order and therefore the order book is taken as it has been downloaded, without any modification. We begin showing in table (\ref{tab:CumRevenueComp}) the cumulative revenue obtained by simulating the execution of each strategy as a percentage of NS, when we liquidate $\Upsilon = 4000$ BNB cryptocurrencies \footnote{we choose a value of $4000$ to ensure that this inventory is liquidatable at a single point in time} in $1800$ seconds at intervals of $5$ seconds.

\begin{table}
	\begin{center}
		\begin{tabular}{|l|l|}
			\hline
			\small
			Strategy & Accumulated \\
			& Revenue (\%) \\
			\hline \hline
			UTLPL & 1.010922 \\ \hline
			UTLPP & 1.009505 \\ \hline
			UTPPL & 1.011206 \\ \hline
			UTPPP & 1.008875 \\ \hline
			ATLPL & 1.008409 \\ \hline
			ATLPP & 1.011454 \\ \hline
			ATPPL & 1.008456 \\ \hline
			ATPPP & 1.011473 \\ \hline
			OTLPL & 1.006268 \\ \hline
			OTLPP & 1.010698 \\ \hline
			OTPPL & 1.007157 \\ \hline
			OTPPP & 1.011098 \\ \hline
			NS & 1.000000 \\ \hline
		\end{tabular}
		\caption{Cumulative revenue as a percentage of NS}
		\label{tab:CumRevenueComp}
	\end{center}
\end{table}

From the results in table \ref{tab:CumRevenueComp}, we observe that interestingly the policy paths that start with a slow trading rate are the ones that earn the highest revenue. The policy path that generated the highest revenue was ATPPP with a cumulative revenue of $1.1473\%$ above NS and in second place was ATLPP with a revenue of $1.1454\%$ above NS. This, although surprising, may be justified by the fact that the number of shares to liquidate gives an average estimation when compared to the LOB. In average-estimation scenario, these policy paths are the ones that start liquidating slowly, even zero, increasing the trading rate until a maximum point and when the inventory is zero, dropping to zero. Similarly, in the overestimation and underestimation scenarios, the policy paths that obtained the highest cumulative revenue were those that started with a low trading rate. 

Likewise, regardless of the scenario, the policy paths that obtained the lowest cumulative revenue were those that liquidated faster at the beginning, such is the case of OTLPL which was the one that obtained the lowest revenue with a $0.6268\%$ above NS and OTPPL, which cumulative revenue was $0.7157\%$ above NS.

This situation can be described in terms of optimal inventory. If an optimal policy path starts by liquidating slowly and then increase the liquidation rate, the optimal inventory path will decrease slowly at first and then decrease more rapidly, i.e., it is a concave function.

In the opposite case, when the policy path liquidates quickly at the beginning and then decreases the liquidation rate, the inventory path will decrease quickly at the beginning and then decrease more slowly, i.e., a convex function.

Then, the optimal inventory path is a non-increasing function, which when concave describes an optimal policy path that, as shown before, generates higher cumulative revenue than convex optimal inventory functions, see figure (\ref{fig:qopt}).

\begin{figure}
	\centering
	{\includegraphics[scale=0.5]{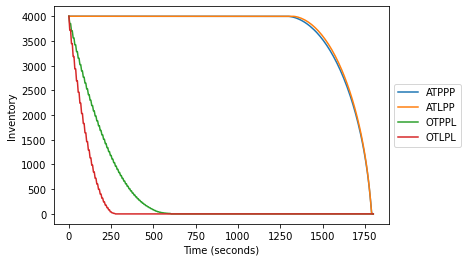}}
	\caption{Optimal inventory paths for which, when concave, such as ATPPP and ATLPP, generate higher cumulative revenues than convex paths, such as OTPPL and OTLPL.} \label{fig:qopt}
\end{figure}

As a way of validating this behavior of the strategies we vary the inventory $\Upsilon$, that is we simulate the liquidation of different inventories starting at $500$ and ending at $4000$ with increments of $500$. The main result of this simulation is shown in figure (\ref{fig:optvcomparqvaries}), in which we observe that as the size of the inventory to be liquidated increases, the percentage of cumulative revenue over NS obtained by optimal strategies also increases. We also note that the strategies with highest revenue are always the same, i.e. ATPPP and ATLPP and the under-performing strategies are also the same, i.e. OTPPL and OTLPL. This is an indication that the performance of the strategy is independent of $\Upsilon$.

\begin{figure}[H]
	\centering
	{\includegraphics[scale=0.5]{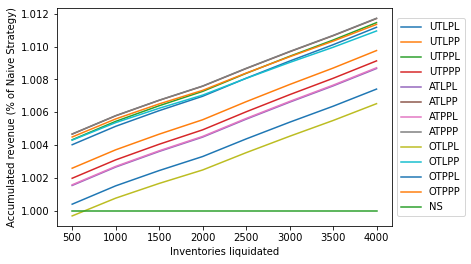}}
	\caption{Comparison of the performance of the optimal strategies through accumulated revenue as a percentage of NS when inventory varies from $500$ to $4000$ with steps of $500$.} \label{fig:optvcomparqvaries}
\end{figure}
To confirm or reject the empirical hypothesis that optimal policy paths that start with a slow liquidation rate and increase as time passes are the ones generating higher cumulative revenue and that, on the contrary, policy paths that start with a high trading rate and decrease as time passes are the ones generating less cumulative revenue, we consider that the optimal inventory paths can be modeled by the non-increasing monotonic function
\begin{equation}
	\begin{array}{lccc}
	q: & \left\{t \in \mathbb{Z}: 0 \leq t \leq T \right\}& \rightarrow &[0,\Upsilon]\\
	 	& t & \rightarrow & d_1 t^{d_2} +d_3
	\end{array}
\end{equation}
subject to $q(0) = \Upsilon$ and $q(T) = 0$, then 
\begin{equation}
	q(t) = \frac{-\Upsilon}{T^{d_2}}t^{d_2} +\Upsilon
\end{equation}
For the parameters $\Upsilon = 4000$, $T=1800$ and taking the exponent $d_2$ in $\left[ 0.1,0.5,1,5,10\right]$ to obtain changes in the concavity of the function $q(t)$, we build simulations of optimal inventory paths, see figure (\ref{subfig:simulq}). With each of these inventory paths, we simulate the execution of the corresponding optimal policy path, thus obtaining its cumulative revenue as a percentage of NS, characterized by $d_2$, see figure (\ref{subfig:cumrevsimulq}).

\begin{figure}
	\centering
	\subfigure[Simulated optimal inventory paths from equation $q(t) = \frac{-4000}{1800^{d_2}}t^{d_2} +4000$]
	{\label{subfig:simulq}\includegraphics[scale=0.38]{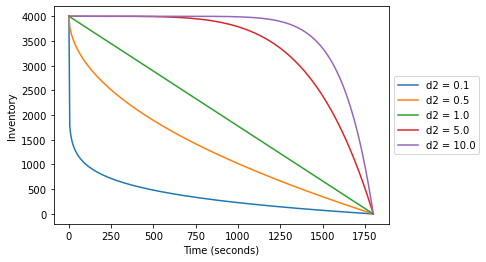}}
	\subfigure[Cumulative revenue from simulated optimal inventory paths]
	{\label{subfig:cumrevsimulq}\includegraphics[scale=0.38]{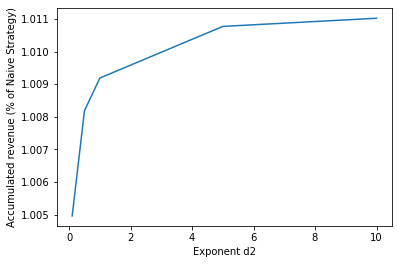}}
	\caption{Simulation of optimal inventory paths and corresponding cumulative revenue as a percentage of NS.} \label{fig:simulq}
\end{figure}

We observe that when the exponent $d_2$ is less than 1, i.e, $q(t)$ is convex, the cumulative revenue as a percentage of NS obtained is less than the other strategies. It seems that by increasing the value of the exponent $d_2$, the cumulative revenue also increases. However, if the value of the exponent $d_2$ is very close to zero or very large, these inventory paths begin to approximate those of the NS strategy and therefore the corresponding cumulative revenue obtained decrease. Figure \eqref{fig:cumrevsimulq} shows the cumulative revenue as a percentage of NS obtained by simulating optimal inventories paths when the exponent $d_2$ varies from $0.01$ to $100$ with increments of $0.1$. We can observe that for values of $d_2$ close to zero, the cumulative revenue is close to that obtained by NS and that for large values of the exponent $d_2$ the cumulative revenue decreases. This situation leads us to conclude that in an interval around $1$ for the exponent $d_2$, the concave optimal inventories paths generate a higher cumulative revenue than convex optimal inventories paths.

\begin{figure}
	\centering
	{\includegraphics[width=3.0in, height=2.6in,keepaspectratio=true]{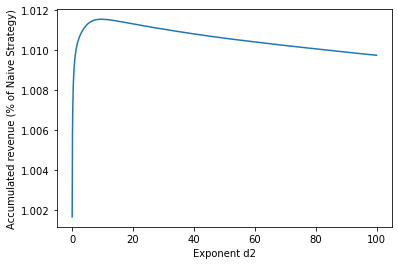}}
	\caption{Cumulative revenue from simulated optimal inventories paths for $d_2$ from $0.01$ to $100$ with steps of $0.1$, liquidating $\Upsilon = 4000$.} \label{fig:cumrevsimulq}
\end{figure}

In terms of optimal policies, those that start with a slow trading rate and increase over time, as opposed to those that start liquidating more quickly and decrease over time, generate higher cumulative revenue. Likewise, we emphasize that this behavior is only observed if such policy paths do not reach the extreme of resembling the NS, since in that case their cumulative revenue decreases.

%------------------------------------------------

\section{Discussion and conclusions}
In this work we study the problem of optimal liquidation by analyzing the consequences of having different functional forms for temporary impact, specifically linear and power, and for permanent impact, specifically linear, quadratic and more general power. We apply this theory in the cryptocurrency market given its importance due to its evident growth, high volatility and its articulation in the world of Decentralized Finance (DeFi) and secondly due to the ease of accessing historical and LOB information. Therefore, we use LOB data for BNB cryptocurrency, information downloaded from the Binance API.

When using LOB data, the depth and therefore the volumes available at the LOB play an important role in the estimation and calibration process of both the functional forms of TPI and PPI. That is, when simulating to liquidate different inventories, if the maximum inventory to be liquidated is relatively small to the inventory available in the data, the impact will be small at the beginning of the simulation, but will increase as the size of the inventories increases, which configures a power functional form ($x^a$ with $a>1$).Additionally,  if the inventory is larger than the available inventory at the LOB, then at the beginning of the simulation there is a greater impact than towards the end of the simulation, because the available inventory will quickly be reached and therefore the impact will decrease, which configures a power functional form ($x^a$, but with $0<a<1$). Likewise, for certain values of inventory in the intermediate of these two extreme cases, the functional form of the impacts tends to linearize ($x^a$ with $a=1$). Future research aims to find specific characteristics in which each of the scenarios is configured, as well as possible analytical models that allow identifying, for example, points at which the impacts are linearized and differences between these.

We observe that the shape of the optimal policy paths can change depending on the starting inventory, the functional forms of TPI and PPI, and the values of the calibrated parameters. However, the functional description of the policy paths does not seem obvious, so we carry out this characterization through the concavity of the optimal inventory paths. Leading us to acknowledge the importance of the ratio between the initial amount of shares to be liquidated and the availability of shares in the LOB when characterizing the performance of optimal policies, and that these do not depend strongly on the functional form of TPI and PPI.

Since this is a problem in which an amount $\Upsilon$ of cryptocurrencies is to be liquidated in a pre-established period of time $T$, the optimal inventory path can be described through a non-increasing monotonic function $\displaystyle q(t) = \frac{-\Upsilon}{T^{d_2}}t^{d_2} +\Upsilon$, which describes a large number of optimal policy paths. Under these conditions, we observe that concave optimal inventory paths, i.e., $d_2>1$, are the ones that generate the greatest cumulative revenue, under the condition that the concavity is not so dramatic as the NS strategy. 

This leads us to characterize the optimal winning policy paths as those whose trading rate is low at the beginning and increases over time. In contrast to those policy paths whose trading rate is high at the beginning and decreases over time, which can be considered as losers, are characterized by convex optimal inventory paths, that is, $0<d_2<1$. 

Further research tends to examine this phenomenon in greater detail, identifying, for example, the exact conditions and points at which the described behavior changes, or, based on families of optimal policies, identifying the ``most'' optimal among them, the functional forms of these families of optimal policies, among other options.

The scope of this study is limited, initially by the LOB data available, by the liquidity of the asset to be liquidated since high liquidity is required so that the inventories can be liquidated, in this case in time horizons of seconds. Other limitations of this study are the model used does not take into consideration other variables that could affect the optimal liquidation policy such as volatility variability, liquidating using limit orders, new market and limit order flows by other agents, the liquidation dynamics that may imply jumps due to ticks, among others. Therefore, considering several additional variables to the problem, different models, different characteristics of financial assets and markets in which they are traded are part of the future research agenda.

\nocite{*}
\bibliographystyle{apalike}
\bibliography{biblio}
\newpage

\appendix

\section{Proof of Lemma \ref{lm:linear}}\label{sec:App1}

In this section we assume that $f(\nu) =a_1\nu+a_2$ and $g(\nu)=b\nu$. The objective is to solve the $\mathcal{H}\mathcal{J}\mathcal{B}$ differential equation (\ref{eqn:AgentHJB}) to find the policy that allows the optimal liquidation of $q_t$ shares in a period of time $[t,T)$. Thus we have to solve
\begin{equation}
\label{eqn:HJBlin}
{\partial }_tH+\frac{1}{2}{\sigma }^2{\partial }_{ss}H +\sup_{\nu \in  \mathcal{A}}\left\{-b \nu\partial_sH-\nu \partial_qH+\left(S-\frac{1}{2}\Delta -a_1\nu-a_2\right)\nu\right\}=0
\end{equation}
Using the first order conditions on $\nu$ we obtain:
\begin{equation}
\label{eqn:DElin}
\partial_tH+\frac{1}{2}{\sigma }^2{\partial }_{ss}H+\frac{\left(S-\frac{1}{2}\Delta -a_2- \partial_q H-b \partial_s H\right)^2}{4a_1}=0 
\end{equation}
Using the educated guess
\begin{equation}
\label{eqn:ansatzLin}
H\left(t,S,q\right)=q\left(S-\frac{1}{2}\Delta-a_2\right)-\frac{b}{2}q^2+h(t,q)
\end{equation}
 and replacing in the differential equation (\ref{eqn:DElin}) we simplify to
\begin{equation}
\label{eqn:DEhlin}
{\partial }_th+\frac{1}{4a_1}{\left({\partial }_qh\right)}^2=0
\end{equation}
thus finding the following solutions
\begin{equation}
\label{eqn:Hli2n}
H(t,s,q)=q\left(S-\frac{1}{2}\Delta-a_2\right)-\left(\frac{b}{2}+\frac{a_1}{(T-t)}\right)q^2
\end{equation}
\begin{equation}
\label{eqn:Nulin2}
\nu^*_t=\frac{q_t}{T-t}
\end{equation}

\section{Proof of Lemma \ref{lm:linqu}}\label{sec:App2}
In this section we assume that $f(\nu) =a_1\nu+a_2$ and $g(\nu)=c_1\nu^2+c_2\nu$. The objective is to solve the $\mathcal{H}\mathcal{J}\mathcal{B}$ differential equation (\ref{eqn:AgentHJB}) to find the optimal policy for liquidating $q_t$ shares in a time period $[t,T)$. Thus we have to solve
\begin{equation}
	\label{eqn:HJBquad}
	{\partial }_tH+\frac{1}{2}{\sigma }^2{\partial }_{ss}H+\sup_{\nu \in  \mathcal{A}}\left\{\left(-c_1 \nu^2-c_2\nu\right)\partial_sH-\nu \partial_qH+\left(S-\frac{1}{2}\Delta -a_1\nu-a_2\right)\nu\right\}=0
\end{equation}
Using the first order conditions on $\nu$ we obtain:
\begin{equation}
	\label{eqn:DEquad}
	\partial_tH+\frac{1}{2}{\sigma }^2{\partial }_{ss}H+\frac{\left(S-\frac{1}{2}\Delta-a_2-\partial_qH-c_2\partial_sH\right)^2}{4\left(c_1{\partial }_sH+a_1\right)}=0
\end{equation}
Using the educated guess
\begin{equation}
	\label{eqn:ansatzQuad}
	H\left(t,S,q\right)=q\left(S-\frac{1}{2}\Delta-a_2\right)-\frac{c_2}{2}q^2+h(t,q)
\end{equation}
and replacing in the differential equation (\ref{eqn:DEquad}) we simplify to
\begin{equation}
	\label{eqn:DE}
	{\partial }_th+\frac{\left({\partial }_qh\right)^2}{4\left(c_1q+a_1\right)}=0 
\end{equation}
thus finding the following solutions
\begin{equation}
	\label{eqn:Hquad2}
	H\left(t,S,q\right)= \left\{ \begin{array}{lcc}
		q\left(S-\frac{1}{2}\Delta-a_2\right)-\frac{c_2}{2}q^2-\frac{4\left(c_1q+a_1\right)^3}{9c_1^2(T-t)} &   if  & q \neq 0 \\
		\\ 0 & if & q = 0
	\end{array}
	\right.
\end{equation}
\begin{equation}
		\label{eqn:Nuquad2}
	\nu_t^*=\frac{2\left(c_1q_t+a_1\right)}{3c_1\left(T-t\right)}
\end{equation}  
\section{Comparison squared residuals and $R^2$}\label{sec:App4}
\begin{table}[H]
	\begin{center}
		\begin{tabular}{|l|l|l|l|l|l|}
			\hline
			\multicolumn{5}{|c|}{Comparison squared residuals $(\hat{ \epsilon}^2 )$}& R-squared $(R^2)$\\ \hline
			impact & model & total sum & mean & std dev & \\
			\hline \hline
			Underestimate TPI (UTPI) & Linear & 0.0000590 & 0.0000012 &  0.0000012 & 0.990627\\ \hline
			Underestimate TPI (UTPI) & Power & 0.0000162 & 0.0000003 & 0.0000006 & 0.997424\\ \hline
			\hline
			Underestimate PPI (UPPI) & Linear & 0.0001789 & 0.0000037 &  0.0000029 & 0.980279\\ \hline
			Underestimate PPI (UPPI) & Power & 0.0000171 & 0.0000003 & 0.0000004 & 0.998111\\ \hline
			\hline \hline
			Overestimate TPI (OTPI) & Linear & 2.4095625 & 0.0491747 &  0.0738702 & 0.980702\\ \hline
			Overestimate TPI (OTPI) & Power & 1.4869937 & 0.0303468 & 0.0374643 & 0.988091\\ \hline
			\hline
			Overestimate PPI (OPPI) & Linear & 5.6188366 & 0.1146701 &  0.1361679 & 0.955715\\ \hline
			Overestimate PPI (OPPI) & Power & 3.8335537 & 0.0782358 & 0.0732665 & 0.969786\\ \hline
			\hline \hline
			Average-estimate TPI (ATPI) & Linear & 0.0010416 & 0.0000213 &  0.0000237 & 0.999803\\ \hline
			Average-estimate TPI (ATPI) & Power & 0.0010396 & 0.0000212 & 0.0000230 & 0.999804\\ \hline
			\hline
			Average-estimate PPI (APPI) & Linear & 0.0096960 & 0.0001979 &  0.0001787 & 0.998061\\ \hline
			Average-estimate PPI (APPI) & Power & 0.0056630 & 0.0001156 & 0.0001251 & 0.998867\\ \hline
		\end{tabular}
		\caption{Comparison of linear and power model fit through squared residuals and $R^2$ for under, over and average estimated TPI and PPI}
		\label{rescompbid}
	\end{center}
\end{table}
\section{Calibrated coeficients}\label{sec:App3}
Since the differential equation \eqref{eqn:AgentHJB} incorporated the direction of the negotiations, i.e. liquidation, through the negative sign for both TPI and PPI, these coefficients must be considered through their additive inverse, except those representing exponents, i.e. $r_2$ and $p_2$.
\begin{table}[H]
	\begin{center}
		\begin{tabular}{|l|l|l|l|l|l|}
			\hline
			\multicolumn{5}{|c|}{Coefficients obtained for each calibrated model}\\ \hline
			impact & model & First parameter & Second parameter & Third parameter  \\
			\hline \hline
			UTPI & Linear & $a_1 = -0.00079754$ & $a_2 = -0.00066177$ &  NA \\ \hline
			UTPI & Power & $r_1 = -2.44114118e-04$ & $r_2 = 1.29520174$ & $r_3 = -3.77323856e-03$ \\ \hline
			\hline
			UPPI & Linear & $b_1 = -0.00095264$ & $b_2 = 0.00229332$ &  NA \\ \hline
			UPPI & Power & $p_1 = -1.38745021e-04$ & $p_2 = 1.48472878$ & $p_3 = -3.01507718e-03$ \\ \hline
			\hline \hline
			OTPI & Linear & $a_1 = -0.00079843$ & $a_2 = -0.22235319$ &  NA \\ \hline
			OTPI & Power & $r_1 = -0.00538481$ & $r_2 = 0.78904313$ & $r_3 = 0.23917224$ \\ \hline
			\hline
			OPPI & Linear & $b_1 = -0.00079455$ & $b_2 = -0.2597109$ &  NA \\ \hline
			OPPI & Power & $p_1 = -0.00947337$ & $p_2 = 0.72704129$ & $p_3 = 0.38622943$ \\ \hline
			\hline \hline
			ATPI & Linear & $a_1 = -0.00095984$ & $a_2 = 0.00209078$ &  NA \\ \hline
			ATPI & Power & $r_1 = -0.0011318 $ & $r_2 = 0.97757467$ & $r_3 = 0.01148375$ \\ \hline
			\hline
			APPI & Linear & $b_1 = -0.0009179$ & $b_2 = -0.01720435$ &  NA \\ \hline
			APPI & Power & $p_1 = -0.00200298$ & $p_2 = 0.89422254$ & $p_3 = 0.02904742$ \\ \hline
		\end{tabular}
		\caption{Calibrated oefficients}
		\label{tab:coeficients}
	\end{center}
\end{table}
%----------------------------------------------------------------------------------------

\end{document}